\documentclass[11pt,letterpaper]{article}
\usepackage[margin=0.85in]{geometry}

\usepackage{comment,multicol}
\usepackage{thm-restate}
\usepackage[ruled,linesnumbered,vlined]{algorithm2e}
\usepackage{multirow}
\usepackage{wrapfig}
\usepackage{color,colortbl}
\usepackage{float}
\usepackage{amsthm}
\usepackage{amsmath}
\usepackage{amssymb}
\usepackage{graphicx}
\usepackage{xspace}
\usepackage{nicefrac}
\usepackage{dsfont}
\usepackage{thmtools}
\usepackage{pdfsync}

\definecolor{Darkblue}{rgb}{0,0,0.4}
\usepackage[colorlinks,linkcolor=Darkblue,filecolor=blue,citecolor=blue,urlcolor=Darkblue,pagebackref]{hyperref}
\usepackage[nameinlink]{cleveref}

\usepackage[toc]{multitoc}

\usepackage[colorinlistoftodos,prependcaption,textsize=tiny]{todonotes}

\usepackage{authblk}

\newtheorem{theorem}{Theorem}
\newtheorem{lemma}{Lemma}
\newtheorem{claim}[lemma]{Claim}

\newtheorem{definition}[lemma]{Definition}

\newtheorem{corollary}[lemma]{Corollary}

\newtheorem{conjecture}{Conjecture}
\numberwithin{equation}{section}

\newcommand{\R}{\mathbb{R}}
\newcommand{\etal}{{et al. \xspace}}
\newcommand{\modu}{\rm mod~}

\date{}

\begin{document}
	\title{Plurality in Spatial Voting Games with Constant $\beta$\thanks{A preliminary version of this paper has appeared in AAAI’21~\cite{FF21}.}
	}
	\author[1]{Arnold Filtser\thanks{This research was supported by the ISRAEL SCIENCE FOUNDATION (grant No. 1042/22).}}
	\author[2]{Omrit Filtser}
	\affil[1]{Department of Computer Science, Bar-Ilan University, Ramat Gan, Israel, \texttt{arnold273@gmail.com}}
	\affil[2]{Department of Mathmatics and Computer Science, The Open University of Israel, Ra'anana, Israel, \texttt{omrit.filtser@gmail.com}}

	\maketitle
	\begin{abstract}
		\sloppy 
		Consider a set $V$ of voters, represented by a multiset in a metric space $(X,d)$. The voters have to reach a decision --- a point in $X$. A choice $p\in X$ is called a $\beta$-plurality point for $V$, if for any other choice $q\in X$ it holds that $|\{v\in V\mid \beta\cdot d(p,v)\le d(q,v)\}|
		\ge\frac{|V|}{2}$. In other words, at least half of the voters ``prefer'' $p$ over $q$, when an extra factor of $\beta$ is taken in favor of $p$.
		For $\beta=1$, this is equivalent to Condorcet winner, which rarely exists. The concept of $\beta$-plurality was suggested by Aronov, de Berg, Gudmundsson, and Horton [TALG 2021] as a relaxation of the Condorcet criterion. 
		
		Let $\beta^*_{(X,d)}=\sup\{\beta\mid \mbox{every finite multiset $V$ in $X$ admits a $\beta$-plurality point}\}$.
		The parameter $\beta^*$ determines the amount of relaxation required in order to reach a stable decision.
		Aronov et al. showed that for the Euclidean plane $\beta^*_{(\R^2,\|\cdot\|_2)}=\frac{\sqrt{3}}{2}$, and more generally, for $d$-dimensional Euclidean space, $\frac{1}{\sqrt{d}}\le \beta^*_{(\R^d,\|\cdot\|_2)}\le\frac{\sqrt{3}}{2}$.
		In this paper, we show that $0.557\le \beta^*_{(\R^d,\|\cdot\|_2)}$ for any dimension $d$ (notice that $\frac{1}{\sqrt{d}}<0.557$ for any $d\ge 4$). In addition, we prove that for every metric space $(X,d)$ it holds that $\sqrt{2}-1\le\beta^*_{(X,d)}$, and show that there exists a metric space for which $\beta^*_{(X,d)}\le \frac12$.

	\end{abstract}
	
	\section{Introduction}
	When a group of agents wants to reach a joint decision, it is often natural to embed their preferences in some metric space. 
	The preferences of each agent are represented by a metric point (also referred to as a voter).
	Each point in the metric space is a potential choice, where an agent/voter prefers choices closer to its point over farther choices.
	The goal is to reach a stable decision, in the sense that no alternative choice is preferred by a majority of the voters. Such a decision is often referred to as a Condorcet winner.

	More formally, consider a metric space $(X,d)$, and a finite multiset of points $V$ from $X$, called voters. A voter $v$ prefers a choice $p\in X$ over a choice $q\in X$ if $d(p,v)< d(q,v)$.
	Specifically, a point $p\in X$ is a \emph{plurality point} if for any other point $q\in X$, 
	the number of voters preferring $p$ over $q$ is at least the number of voters preferring $q$ over $p$, i.e., $\left|\left\{v\in V\mid d(p,v)< d(q,v)\right\}\right|\ge\left|\left\{v\in V\mid d(p,v)> d(q,v)\right\}\right|$.\footnote{A more accurate name for such a point, which is also used in the literature, is \emph{Condorcet winner}.
	However, as this work is mainly concerned with the term $\beta$-plurality point defined in \cite{ABGH21}, we choose to keep their terminology. 		
	}
	
	The special case where $(X,d)$ is the Euclidean space, i.e., $(\R^d,\|\cdot\|_2)$, is called spatial voting games, and was studied in the political economy context \cite{black1948rationale,downs1957economic,Plott67,EH83}.
	When $X=\R$ is the real line, a plurality point always exist, in fact, it is simply the median of $V$ (for even $V$, there are two plurality points). When $(X,d)$ is induced by the shortest path metric of a tree graph, then again a plurality point always exists, as any separator vertex\footnote{If $T$ is the tree inducing $(X,d)$, a separator vertex is a vertex $z\in X$, the removal of which will break the graph $T$ into connected components, each containing at most $\frac{|V|}{2}$ voters. Every tree contains a separator vertex \cite{JOR69}.} is a plurality point.
	However, already in $\R^2$ a plurality point does not always exist, and moreover, it exists only for a negligible  portion of the point sets.
	Indeed, for any $d\ge2$, a plurality point for a multiset $V$ in $\R^d$ exists if and only if all median hyperplanes\footnote{A median hyperplane for $V$ is a hyperplane such that both open half-spaces defined by it contain less than $\frac{|V|}{2}$ voters.\label{foot:medianHyper}} for $V$ have a common intersection point (see \cite{EH83,Plott67}). Wu et al. \cite{WLWC14} and de Berg et al. \cite{BGM18} presented algorithms that determine whether such a point exist.

	Recently, Aronov, de Berg, Gudmundsson, and Horton \cite{ABGH21}, introduced a relaxation for the concept of plurality points, by defining a point $p\in X$ to be a \emph{$\beta$-plurality point}, for $\beta\in(0,1]$, if for every other point $q\in X$, it holds that $\left|\left\{v\in V\mid \beta\cdot d(p,v)< d(q,v)\right\}\right|\ge\left|\left\{v\in V\mid \beta\cdot d(p,v)> d(q,v)\right\}\right|$. In other words, if we scale distances towards $p$ by a factor of $\beta$, then for every choice  $q$, the number of voters preferring $p$ over $q$ is at least the number of voters preferring $q$ over $p$. Set 
	\begin{align}
		\beta_{(X,d)}(p,V):= & \sup\{\beta\mid p\mbox{ is a }\beta\mbox{-plurality point in }X\mbox{ w.r.t. }V\}~,\nonumber\\
		\beta_{(X,d)}(V):= & \sup_{p\in X}\{\beta_{(X,d)}(p,V)\}~,\label{eq:AdBGH21Def}\\
		\beta_{(X,d)}^{*}:= & \inf\{\beta_{(X,d)}(V)\mid\mbox{\ensuremath{V} is a multiset in \ensuremath{X}}\}~.\nonumber
	\end{align}
	A natural question is to find or estimate these parameters for a given metric space.
	Notice that as $\beta$ becomes larger, 
	a $\beta$-plurality point becomes more similar to a ``normal'' plurality point, and for $\beta=1$ the two concepts are the same.
	Therefore, it is interesting to know what values of $\beta$ are required for a given metric space in order to reach a stable decision. These bounds give an indication on the amount of relaxation that might be needed, and how reasonable it is.

	Aronov \etal \cite{ABGH21} studied $\beta$-plurality for the case of Euclidean space, i.e., $(\R^d,\|\cdot\|_2)$. Given a specific instance $V$, they presented an EPTAS to approximate $\beta_{(\R^d,\|\cdot\|_2)}(V)$.
	For the case of the Euclidean plane ($d=2$), they showed that $\beta^*_{(\R^2,\|\cdot\|_2)}= \frac{\sqrt{3}}{2}$. 
	Specifically, they showed that for every multiset of voters $V$ in $\R^2$, there exists a point $p\in \R^2$ such that $\beta_{(\R^2,\|\cdot\|_2)}(V,p)\ge\frac{\sqrt{3}}{2}$.
	Furthermore, they showed that for the case where $V$ consist of the three vertices of an equilateral triangle, it holds that $\beta_{(\R^2,\|\cdot\|_2)}(V)\le \frac{\sqrt{3}}{2}$. 
	For the general $d$-dimensional Euclidean space  $(\R^d,\|\cdot\|_2)$, Aronov \etal showed a lower bound of $\beta_{(\R^2,\|\cdot\|_2)}^*(V)\ge \frac{1}{\sqrt{d}}$. 
	The problem of closing the gap between $\frac{1}{\sqrt{d}}$ and $\frac{\sqrt{3}}{2}$ was left by Aronov \etal as a ``main open problem''. In addition, they asked what bound on $\beta^*$ could be proved in other metric spaces.
	
	\begin{table}[t]
		\centering
		\begin{tabular}{|l|l|l|l|}
			\hline
			\textbf{Space}                                          & \textbf{Lower Bound} & \textbf{Upper Bound}  & \textbf{Ref}                                                                                                  \\ \hline\hline
			\textbf{$\mathbb{R}$ and tree metric}                    & 1                    & 1                    &                                                                                                               \\ \hline
			\textbf{$(\mathbb{R}^2,\|\cdot\|_2)$}                    & $\nicefrac{\sqrt{3}}{2}\approx0.866$ & $\nicefrac{\sqrt{3}}{2}$  & \cite{ABGH21}                                                                                \\ \hline
			\textbf{$(\mathbb{R}^3,\|\cdot\|_2)$}& $\nicefrac{1}{\sqrt{3}}\approx 0.577$ & $\nicefrac{\sqrt{3}}{2}$  & \cite{ABGH21}                                                                                \\ \hline
			\textbf{$(\mathbb{R}^d,\|\cdot\|_2)$ for $d\ge 4$}       & $\approx0.557$  & $\nicefrac{\sqrt{3}}{2}$         &\Cref{thm:Euclidean}, \cite{ABGH21}                                            \\ \hline
			\textbf{General metric space}                           & $\sqrt{2}-1\approx 0.414$  & $\nicefrac12$                    &\Cref{thm:main}, \Cref{thm:LB} \\ \hline
		\end{tabular}
		\caption{\label{tbl:results}Summary of current
			and previous results on $\beta^*_X$ for different metric spaces.}
		\vspace{-10pt}
	\end{table}

	\paragraph{Our contribution.} We prove that for every metric space $(X,d)$, it holds that $\beta^*_{(X,d)}\ge\sqrt{2}-1$. Note that Aronov et al.~\cite{ABGH21} gave a lower bound of $\frac{1}{\sqrt{d}}$ for the Euclidean metric, while our result shows a constant lower bound for any metric space.
	In addition,  we provide an example of a metric space $(X,d)$ for which $\beta^*_{(X,d)}= \frac12$. In fact, we show that $\beta^*_{(X,d)}\le \frac12$ for any (continuous) graph metric $(X,d)$ that contains a cycle (in contrast to tree metrics, for which $\beta^*_{(X,d)}= 1$). 
	Finally, for the case of Euclidean space of arbitrary dimension $d$, we show that $\beta^*_{(\R^d,\|\cdot\|_2)}\ge0.557$. 
	Note that this lower bound is larger than $\frac{1}{\sqrt{d}}$ for $d\ge 4$.
	All the current and previous results are summarized in \Cref{tbl:results}.

	\paragraph{Related work}
	A well known relaxation for the concept of plurality points in Euclidean space is the yolk ~\cite{Mckelvey86,MGF89,FGM88,GW19,Miller2015}, which is the smallest ball intersecting every median hyperplane~$^{\ref{foot:medianHyper}}$ of $V$. The center of the yolk is a good heuristic for a plurality point (see \cite{MG08} for a list of properties the yolk posses).
	Notice that the definition of $\beta$-plurality applies for any metric space, not necessarily Euclidean as in the concept of yolk.
	
	Another relaxation studied by Lin \etal\cite{LWWC15} is the ``minimum cost plurality problem''. Here, given a set of voters $V$ with some cost function, the goal is to find a set $W$ of minimum cost such that $V\setminus W$ contains a plurality point. They prove that the problem is NP-hard in general. 
	de Berg, Gudmundsson, and Mehr \cite{BGM18} (improving over \cite{LWWC15}) provided an $O(n^4\cdot d^2)$ time algorithm for the case of equal costs.
	
	A main drawback of the spatial voting model in the realistic political context was underlined by Stokes~\cite{Stokes63}. The claim is that this model does not take into account the so-called ``valence issues'': qualities of the candidates such as charisma and competence~\cite{ECGV18}, a strong party support~\cite{Wiseman06}, and even the campaign spending~\cite{Herrera08}.	
	Therefore, several more realistic models have been proposed (see, e.g., \cite{GR19,GHR11,SCSW11}). A common model is the multiplicative model which was introduced by Hollard and Rossignol~\cite{HR08}, and is defined for two-candidate spatial voting model. This model is closely related to the concept of $\beta$-plurality. 
	In more detail, in \cite{HR08} there are two candidates with given valences $\sigma_1,\sigma_2$, and they need to choose policies $x_1,x_2\in \R^d$, respectively. A voter $a\in\R^d$ prefers the first candidate if $\frac{1}{\sigma_1}\cdot\|x_1-a\|_2<\frac{1}{\sigma_2}\cdot\|x_2-a\|_2$ (and the second if $\frac{1}{\sigma_1}\cdot\|x_1-a\|_2>\frac{1}{\sigma_2}\cdot\|x_2-a\|_2$). In their study, the valences are fixed, and the candidates are choosing polices in order to win the election. However, the information about the preferred policies of the voters (points in $\R^d$) is not given in full (e.g., only their distribution is given). In contrast, in our paper all the information about the voters is known, and we are looking for the minimum valence that will allow a candidate to choose a single policy, so that he will win the election against any other policy represented by a candidate with valence $1$.
	
	\section{General Metric Spaces}\label{sec:generalMetric}
	We begin by providing a (slightly) alternative definition of $\beta$-plurality point.
	\begin{definition}\label{def:Plurpoint}
		Consider a metric space $(X,d)$, and a multiset $V$ in $X$ of voters. A point $p\in X$ is a $\beta$-plurality point if for every $q\in X$, we have $\left|\left\{v\in V\mid \beta\cdot d(p,v)\le d(q,v)\right\}\right|\ge\frac{|V|}{2}$.
		The rest is similar to \cite{ABGH21} (and \cref{eq:AdBGH21Def}): $\beta_{(X,d)}(p,V)= \sup\{\beta\mid p\mbox{ is a }\beta\mbox{-plurality point in }X\mbox{ w.r.t. }V\}$, 	
		$\beta_{(X,d)}(V)=\sup_{ p\in X}\beta(p,V)$, and $\beta^*_{(X,d)}=\inf\{\beta_{(X,d)}(V)\mid \mbox{$V$ is a multiset in $X$} \}$.
	\end{definition}
	The difference between the definitions is that \Cref{def:Plurpoint} is deciding ties in favor of $p$, that is, a voter $v$ for which $\beta\cdot d(p,v)=d(q,v)$, will choose $p$ over $q$, while in the original definition from \cite{ABGH21}, such voters remain ``undecided''.
	The $\beta_{(X,d)}(p,V)$ parameter is equivalent in these two definitions. This happens due to the supremum used in the definitions which eliminates the difference between strong and weak inequalities ($<,\le$). We prove this equivalence formally in \Cref{sec:defEquiv}. 
	
	Consider a metric space $(X,d)$, with a multiset $V$ of voters from $X$, and set $|V|=n$.
		For a point $p$ and radius $r$, denote by $B_V(p,r)=\{v\in V\mid d(p,v)\le r\}$ the (multi) subset of voters at distance at most $r$ from $p$ (i.e., those that are contained in the closed ball of radius $r$ centered at $p$), and let $R_p$ be the minimum radius such that $|B_V(p,R_p)|\ge\frac{n}{2}$.
	
	The following theorem shows that a $(\sqrt{2}-1)$-plurality point always exists. The fact that the lower bound is constant, and even close to $\frac12$, demonstrates the strength of $\beta$-plurality in the sense that for any set of voters and in any metric space, the multiplication factor needed for the existence of such winner is a fixed constant, and thus the amount of relaxation is bounded.
	\begin{theorem}\label{thm:main}
		For every metric space $(X,d)$, we have $\beta^*_{(X,d)}\ge\sqrt{2}-1$.
	\end{theorem}
	\emph{Proof.}~\sloppy Let $p^*\in X$ be the point with minimum $R_{p}$ over all $p\in X$, and let $B_{p^*}=B_V(p^*,R_{p^*})$.
		We claim that $p^*$ is a $(\sqrt{2}-1)$-plurality point. 
		
		\begin{wrapfigure}{r}{0.3\textwidth}
			\begin{center}
				\vspace{-35pt}
				\includegraphics[scale=0.5]{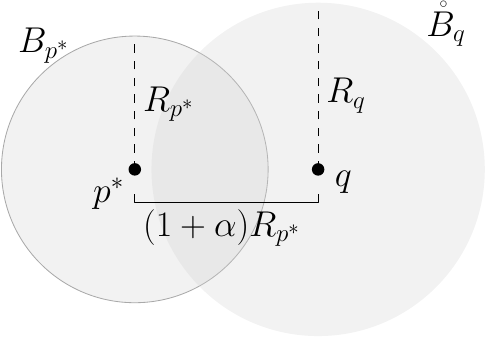}
				\vspace{-15pt}
			\end{center}
			\vspace{-30pt}
		\end{wrapfigure}
	
		Set $\beta=\sqrt{2}-1$, and notice that $\beta=\frac{1}{2+\beta}$. 
		Consider some choice $q\in X$, and let
		$\alpha\ge-1$ be such that $d(p^*,q)=(1+\alpha)\cdot R_{p^*}$. Let $\mathring{B}_q=\{v\in V\mid d(q,v)< R_q\}$ be the (multi) subset of voters at distance (strictly) smaller than $R_q$ from $q$ (i.e., those that are contained in the open ball of radius $R_q$ centered at $q$).
		Consider the following cases:
		\begin{itemize}
			
			\item $\boldsymbol{\alpha\le\beta}$: For every point $v\notin \mathring{B}_q$, as $d(q,v)\ge R_q\ge R_{p^*}$, by the triangle inequality it holds that $$d(p^*,v)\le d(p^*,q)+d(q,v)\le (2+\alpha)\cdot d(q,v)\le (2+\beta)\cdot d(q,v)=\frac1\beta\cdot d(q,v)~.$$ 
			
			\item $\boldsymbol{\alpha\ge\beta}$: For every point $v\in B_{p^*}$, as $d(p^*,q)=(1+\alpha)\cdot R_{p^*}\ge (1+\alpha)\cdot d(p^*,v)$, it holds that $$d(q,v)\ge d(q,p^*)-d(p^*,v)\ge (1+\alpha-1)\cdot d(p^*,v)\ge\beta\cdot d(p^*,v)~.$$
		\end{itemize}
		The theorem follows as $|\mathring{B}_q|< \frac n2\le |B_{p^*}|$.
	\qed
	
	\begin{theorem}\label{thm:LB}
		There exist a metric space $(X,d)$ such that $\beta^*_{(X,d)}= \frac12$.
	\end{theorem}

	\emph{Proof.}~
	Consider the metric space $(X,d)$, where $X$ denotes the set of points on a circle of perimeter $1$ and distances are measured along the arcs. More formally,
	$X$ is the segment $[0,1)$, and given two points $x,y\in [0,1)$, their distance is $d(x,y)=\min\{(x-y)\modu 1,(y-x)\modu 1\}$.
	
	First we show that $\beta^*_{(X,d)}\le \frac12$.
	Consider a set of three voters $\{v_1,v_2,v_3\}=\{0,\frac13,\frac23\}$, all at distance $\frac13$ from each other. We will show that $\beta^*_{(X,d)}(V)\le \frac12$.
	Assume by contradiction that there is a choice $p$ which is a $\beta$-plurality point for $\beta>\frac12$.
	Assume w.l.o.g. that $p=\alpha\in[0,\frac16]$ (see the figure below for illustration), the other cases are symmetric.
	
	\begin{wrapfigure}{r}{0.2\textwidth}
		\begin{center}
			\vspace{-30pt}
			\includegraphics[scale=0.4]{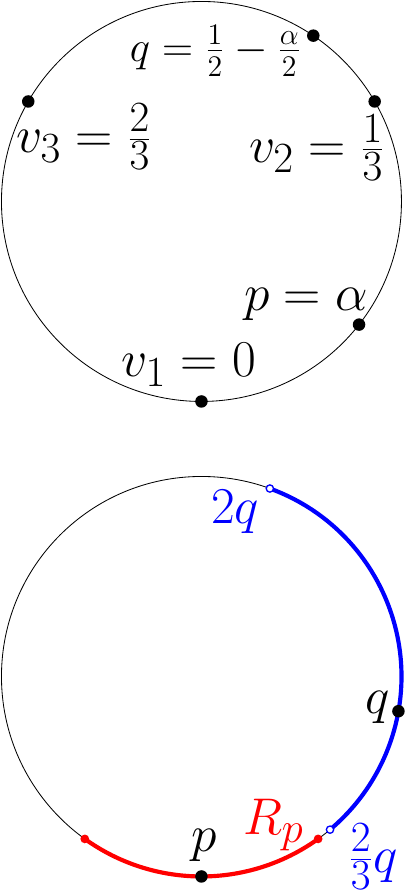}
			\vspace{-5pt}
		\end{center}
		\vspace{-15pt}
	\end{wrapfigure}	

	Consider the choice  $q=\frac12-\frac{\alpha}{2}$ lying on the arc $[v_2,v_3]$ at distance $\frac16-\frac{\alpha}{2}$ from $v_2$, and $\frac16+\frac\alpha2$ from $v_3$. Then $\beta\cdot d(p,v_2)=\beta\cdot (\frac13-\alpha)> \frac16-\frac\alpha2= d(q,v_2)$ and ${\beta\cdot d(p,v_3)=\beta\cdot (\frac13+\alpha)> \frac16+\frac\alpha2= d(q,v_3)}$, which contradicts the assumption that $p$ is a $\beta$-plurality point.
	
	Next we show that $\beta^*_{(X,d)}\ge \frac12$. Consider an arbitrary (multi) subset of voters $V\subseteq X$, and let $p\in X$ be the choice with minimal radius $R_p$ such that $|B_V(p,R_p)|\ge\frac n2$.
	Note that the length of the smallest arc containing $\frac n2$ voters is $2R_p$. In particular, by averaging arguments $2R_p\le \frac12$, and thus $R_p\le\frac14$. 
	Assume w.l.o.g. that $p=0$. We show that $p$ is a $\frac12$-plurality point. Let $q\in X$ be any other point. We assume that $q\in(0,\frac12]$, the case $q\in [\frac12,1)$ is symmetric. 
	Let $v$ be a voter that prefers $q$ over $p$, then $\frac12 d(p,v)>d(q,v)$. 
	If $v<q$ then $\frac12 d(p,v)=\frac12 v$ and $d(q,v)=q-v$, and thus $v>\frac23q$. Else, we have $v>q$, and so $\frac12 d(p,v)\le \frac12 v$ and $d(q,v)=v-q$ (as otherwise the shortest path from $v$ to $q$ goes through $p$, implying $d(p,v)<d(q,v)$), and therefore $v<2q$.
	We conclude that only voters in the arc $(\frac23q,2q)$ prefer $q$ over $p$. The rest is case analysis:
	\begin{itemize}
		\item If $q<\frac32 R_p$, then the arc containing the set of the voters preferring $q$ over $p$ is of length $\frac43 q<2R_p$. By the minimality of $R_p$, it contains less than $\frac n2$ voters.
		\item If $q\ge\frac32 R_p$, then the arc $[0,R_p]$ is disjoint from the arc $(\frac23q,2q)$. Moreover, as $q<\frac12$, all the voters in the arc $[1-R_p,1)\subseteq [\frac34,1)$ will prefer $p$ over $q$. 
		In particular there are at least $\frac n2$ voters preferring $p$ over $q$. 
	\end{itemize}
    \qed
\newline

Given a weighted graph $G=(V,E,w)$, consider its \emph{continuous counterpart}, denoted $\tilde{G}$: each edge $e=(v,u)$ in $G$ is represented in $\tilde{G}$ by a an interval of length $w(e)$, equipped with the line metric with endpoints $u,v$. 
The distance between two points $u,v\in\tilde{G}$, denoted $d_{\tilde{G}}(u,v)$, is the shortest length of a geodesic path connecting $u$ to $v$.

If $G$ contains a cycle, then we can generalize \Cref{thm:LB} to $\tilde{G}$.
This shows a separation between metric spaces obtained by acyclic graphs (trees) which always contain a plurality point (that is, $\beta^*_{(X,d)}= 1$), and metric spaces obtained by all other graphs, for which $\beta^*_{(X,d)}\le\frac12$.
\begin{restatable}{theorem}{LBgraphs}
	\label{thm:LBgenerlized}
	For every weighted graph $G=(V,E,w)$ containing a cycle, it holds that $\beta^*_{(\tilde{G},d_{\tilde{G}})}\le \frac12$.
\end{restatable}
\begin{proof}
	Let $C$ be a cycle in $G$ of minimum length. Assume w.l.o.g. that the length of $C$ is $1$. We place $3$ voters $v_1,v_2,v_3$ on $\tilde{C}$ at equal distance (of $\frac13$) from each other.
	Assume by contradiction that there is a choice $p$ which is a $\beta$-plurality point for $\beta>\frac12$. Furthermore, assume w.l.o.g. that $v_1$ is the voter closest to $p$, and let $d_{\tilde{G}}(p,v_1)=\alpha$. We proceed by case analysis, see \Cref{fig:continuesCounterpart} for an illustration:
	\begin{figure}[h]
		\centering
		\includegraphics[scale=0.72]{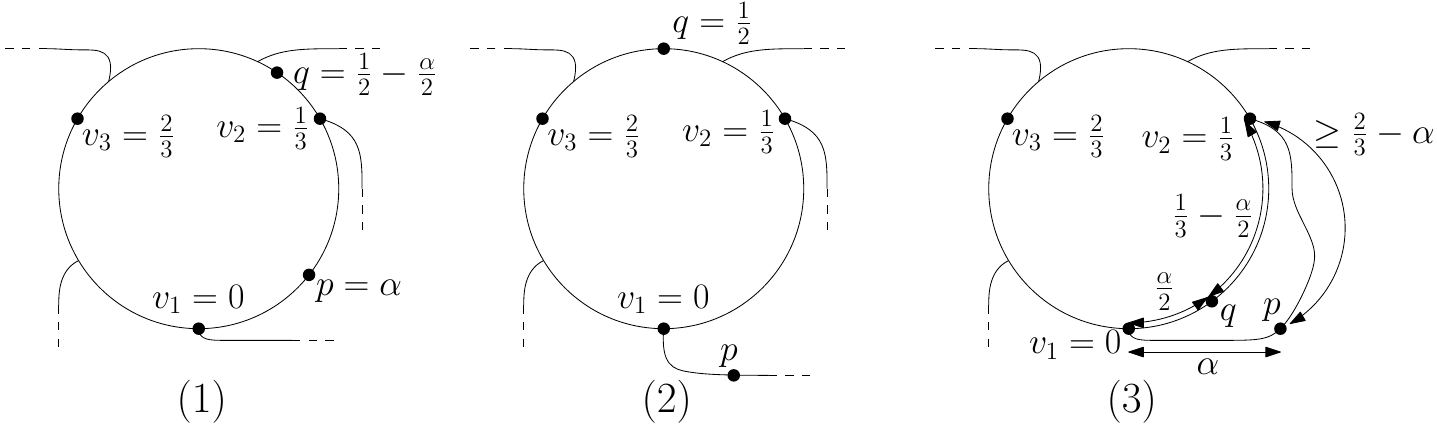}
		\caption{Illustration of the three cases in the proof of \Cref{thm:LBgenerlized}.}
		\label{fig:continuesCounterpart}
	\end{figure}
	\begin{itemize}
		\item \textbf{Case (1).} $p$ lies on the cycle $\tilde{C}$. As $C$ is a cycle of minimum length in $G$, $\tilde{C}$ contains the shortest paths between all $v_1,v_2,v_3$ in $\tilde{G}$ (otherwise there would've been a shorter cycle).
		Following the same argument as in \Cref{thm:LB}, for every possible placement of $p$, there is a choice $q\in \tilde{C}$ that will win over $2$ voters, a contradiction.
		
		\item \textbf{Case (2).}  $p\notin \tilde{C}$, and $v_1$ lies on the shortest paths from $p$ to both $v_2,v_3$.
		Then we have $d_{\tilde{G}}(p,v_2)\ge \frac13$ and $ d_{\tilde{G}}(p,v_3)\ge \frac13$. Consider the choice $q$ lying at distance $\frac16$ from both $v_2,v_3$. Then $q$ will win two voters over $p$, a contradiction.
		
		\item \textbf{Case (3).}  $p\notin \tilde{C}$, and $v_1$ does not lies on the shortest paths from $p$ to both $v_2,v_3$. Suppose  w.l.o.g. that the shortest path from $p$ to $v_2$ does not go through $v_1$, and let $\kappa=d_G(p,v_2)$.
		Since $v_1$ is the voter closest to $p$, there are two different paths in $\tilde{G}$ from $p$ to $v_2$ of lengths $\kappa$ and $\frac13+\alpha$. In particular, $G$ contains a cycle of length at most $\kappa+\frac13+\alpha$. As $C$ is the minimum cycle in $G$, and it is of length $1$, it follow that $\kappa\ge \frac23-\alpha$.
		Let $q$ be the point on $C$ at distance $\frac\alpha2$ from $v_1$ and $\frac13-\frac\alpha2$ from $v_2$. Note that $q$ wins both the votes of $v_1$ and $v_2$ over $p$, a contradiction.
	\end{itemize}

\end{proof}

	\section{Euclidean Space}
	In this section we consider the case of the Euclidean metric space, and give a bound on $\beta^*_{(\R^d,\|\cdot\|_2)}$ which is independent of $d$ and greater than $\frac{1}{\sqrt{d}}$ for any $d\ge 4$, thus improving the lower bound of \cite{ABGH21} for $d\ge 4$.
	As this entire section deals only with Euclidean space, in order to simplify notation, in this section (and the related \Cref{appballObservation,app:ChooseBeta}) we will drop the subscript from   $\|\cdot\|_2$ (writing $\|\cdot\|$), and from $B_{\R^d}(\vec{x},r)$ (writing $B(\vec{x},r)$).
	\begin{theorem}\label{thm:Euclidean}
		For Euclidean space of arbitrary dimension,  $\beta^*_{(\R^d,\|\cdot\|_2)}\ge\beta$,\\\phantom{a}\hfill for $\beta=\frac{1}{2}\sqrt{\frac{1}{2}+\sqrt{3}-\frac{1}{2}\sqrt{4\sqrt{3}-3}}\approx0.55701571813579904605525266098621644838064149582041992...$
	\end{theorem}
	
	We begin with the following structural observation regarding the Euclidean space. 
	\begin{claim}\label{clm:ballObservation}
		Fix a pair of choices $\vec{a},\vec{b}\in\mathbb{R}^d$.
		For any $\beta\in(0,1)$, the set of all voters $\vec{v}\in V$ that do not $\beta$-prefer $\vec{a}$ over $\vec{b}$, i.e., $\left\{ \vec{v}\in V \mid \beta\cdot\left\| \vec{a}-\vec{v}\right\| >\| \vec{b}-\vec{v}\| \right\}$, is contained in the open ball centered at $\vec{o}=\vec{a}+\frac{1}{1-\beta^{2}}\cdot(\vec{b}-\vec{a})$ with radius $\beta\cdot \|\vec{o}-\vec{a}\|$.
	\end{claim}
	This claim was previously known (in fact, the ball is the ``bisector'' used in a multiplicatively weighted Voronoi diagram, see \cite{AE84}). For completeness, we provide  a full proof of the claim in \Cref{appballObservation}, using the notations of our paper.
	
	By the above claim we can conclude:
	\begin{corollary}\label{cor:ballObservation}
		For any $\beta\in(0,1)$, $\vec{a}$ is a $\beta$-plurality point if and only if, for every other point $\vec{o}\in\R^d$, the open ball of radius $\beta\cdot\|\vec{o}-\vec{a}\|$ around $\vec{o}$ contains at most $\frac n2$ voters. 
	\end{corollary}
	
	For the remainder of the section, $\beta$ is the number defined in \Cref{thm:Euclidean} and not a general parameter.
	\begin{proof}[Proof of \Cref{thm:Euclidean}]
		Consider a multiset $V\subseteq \R^{d}$ of voters. Let $\vec{p}$ be the point that minimizes $R_{\vec{p}}$. By scaling and shifting, we can assume w.l.o.g. that $R_{\vec{p}}=1$ and $\vec{p}=\vec{0}$.
		If $\vec{p}$ is a $\beta$-plurality point, then we are done.
		Otherwise, by \Cref{cor:ballObservation} there is a point $\vec{q}$ such that the open ball $B\left(\vec{q},\beta\cdot\|\vec{p}-\vec{q}\| \right)$ contains strictly more than $\frac n2$ voters. Let $q=\|\vec{q}\|$.
		Set $\vec{w}= \left(\frac12(1-\beta^2)q-\beta+\frac{3}{2q}\right)\cdot\frac{\vec{q}}{\|\vec{q}\|}$. We claim that $\vec{w}$ is a $\beta$-plurality point.
		
		First, notice that $q>\frac1\beta$, as otherwise the open ball of radius $\beta q\le 1$ around $\vec{q}$ contains more than $\frac n2$ voters, a contradiction to the fact that $R_{\vec{p}}=1$ is the minimum radius of a closed ball containing at least $\frac n2$ voters.
		Second, it must hold that $q<\frac{1}{1-\beta}$, because otherwise
		$\beta q+1\le q$, implying that the ball $B(\vec{p},R_{\vec{p}})$ and the open ball $B(\vec{q},\beta q)$ are disjoint, a contradiction to the fact that the open ball $B(\vec{q},\beta q)$ contains more than $\frac n2$ voters.
		Therefore, we conclude that
		\begin{equation}
		\frac{1}{\beta}< q <\frac{1}{1-\beta}\label{eq:pq}
		\end{equation}
		
		Notice that $\vec{p}$ is a $\frac12$-plurality point. Indeed, we could've fixed $\beta=\frac12$ and have the exact same discussion leading to \cref{eq:pq}. However, as no $q$ satisfies $\frac12<q<\frac12$, it follows that  $\vec{p}$ is a $\frac12$-plurality point.
		
		To prove that $\vec{w}$ is a $\beta$-plurality point, we will show that for every other point $\vec{z}\in\R^d$, the open ball of radius $\beta\cdot\|\vec{z}-\vec{w}\|$ around $\vec{z}$ contains at most $\frac n2$ voters. 
		We will use the following lemma.
				
		\begin{lemma}\label{lem:wWinner}
			For any point $\vec{z}\in \R^d$, let $z=\|\vec{z}-\vec{w}\|$. Then at least one of the following hold:
			\vspace{-5pt}
			\begin{enumerate}
				\item $z\le\frac1\beta$.
				\item $\|\vec{z}-\vec{p}\|\ge1+\beta z$.
				\item $\|\vec{z}-\vec{q}\|\ge \beta q+\beta z$.
			\end{enumerate}
		\end{lemma}
	
		Before proving \Cref{lem:wWinner}, we show how it implies that $\vec{w}$ is a $\beta$-plurality point.
		For any $\vec{z}\in \R^d$: 
		\begin{itemize}
			\item If $z\le\frac1\beta$, then $\beta z\le 1= R_{\vec{p}}$, and thus the open ball $B(\vec{z},\beta z)$ contains at most $\frac{n}{2}$ voters.
			\item If $\|\vec{z}-\vec{p}\|\ge1+\beta z$, then the balls $B(\vec{p},1)$ and $B(\vec{z},\beta z)$ are disjoint, and thus $B(\vec{z},\beta z)$ contains at most $\frac n2$ voters.
			\item If $\|\vec{z}-\vec{q}\|\ge \beta q+\beta z$, then the balls $B(\vec{q},\beta q)$ and $B(\vec{z},\beta z)$ are disjoint, 
			and thus $B(\vec{z},\beta z)$ contains at most $\frac n2$ voters.		
		\end{itemize} 
		We conclude that for every $\vec{z}\in\R^d$, the ball $B(\vec{z},\beta\cdot z)$ contains at most $\frac n2$ voters, and thus by \Cref{cor:ballObservation}, $\vec{w}$ is a $\beta$-plurality point.
	\end{proof}
	
	\begin{figure}[h]
		\centering
		\includegraphics[page=3,scale=0.8]{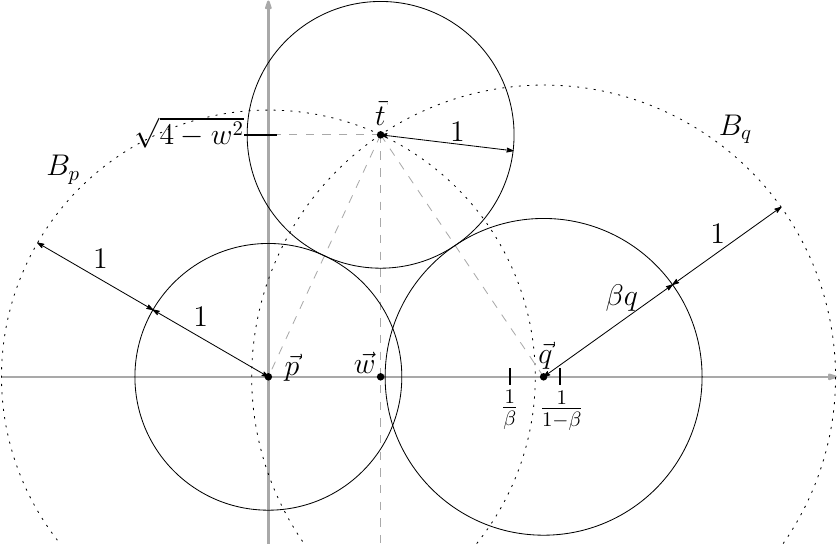}
		\caption{The points $\vec{p}=(0,0)$, $\vec{q}=(q,0)$, and $\vec{w}=(w,0)$ for $w=\frac12(1-\beta^2)q-\beta+\frac{3}{2q}$ are on the $x$-axis. $B_p$ is the circle of radius $2$ around $\vec{p}$, while $B_q$ is the circle of radius $1+\beta q$ around $\vec{q}$. The balls $B_p$ and $B_q$ intersect at $\vec{t}=(w,\sqrt{4-w^2})$ and $\vec{t'}=(w,-\sqrt{4-w^2})$.
			The ball of radius $1$ around $\vec{t}$ is tangent to both $B_p$ and $B_q$.
			It holds that $\|\vec{w}-\vec{t}\|=\sqrt{4-w^2}\le\frac1\beta$ (equation (\ref{eq:ChoiseBeta})).}
		\label{fig:euclidean}
	\end{figure}
	\begin{proof}[Proof of \Cref{lem:wWinner}]
		The points $\vec{p},\vec{q},\vec{w}$ lie on a single line. Given an additional point $\vec{z}$, the four points lie on a single plane. Thus, w.l.o.g. we can restrict the analysis to the Euclidean plane. Moreover, we can assume that $\vec{p}=(0,0)$, $\vec{q}=(q,0)$, $\vec{w}=(w,0)$ for $w=\frac12(1-\beta^2)q-\beta+\frac{3}{2q}$, and that $\vec{z}=(z_x,z_y)$ where $z_y\ge 0$ (the case of $z_y\le 0$ is symmetric).
		
		Let $B_p=B(\vec{p},2)$ and $B_q=B(\vec{q},1+\beta q)$ (see \Cref{fig:euclidean}). 
		The boundaries of  $B_p$ and $B_q$ intersect at the points $(w,\pm\sqrt{4-w^2})$ (this is the reason for our choice of $w$).
		Let $\vec{t}=(w,\sqrt{4-w^2})$, and notice that $0<w<q$ for any $q\ge \frac{1}{\beta}$ (this can be verified by straightforward calculations).
		\Cref{lem:wWinner} follows by the two following claims:
		\begin{claim}\label{clm:inBpBq}
			If $\vec{z}\in B_p\cap B_q$ then $\|\vec{z}-\vec{w}\|\le\frac1\beta$.
		\end{claim}
		\begin{claim}\label{clm:Bp}
			If $\vec{z}\notin B_p\cap B_q$ then either $\|\vec{z}-\vec{p}\|\ge1+\beta z$
			or $\|\vec{z}-\vec{q}\|\ge \beta q+\beta z$.
		\end{claim}
	\end{proof}
	
	\vspace{5pt}
	\emph{Proof of \Cref{clm:inBpBq}.}~
	The boundaries of  $B_p$ and $B_q$ intersect at the points $\vec{t}=(w,\sqrt{4-w^2})$ and $\vec{t'}=(w,-\sqrt{4-w^2})$.
	We claim that for every $q\in(\frac1\beta,\frac{1}{1-\beta})$, it holds that
	\begin{equation}
	\|\vec{t}-\vec{w}\|=\sqrt{4-w^{2}}\le\frac1\beta~.\label{eq:ChoiseBeta}
	\end{equation}
	In fact, $\beta$ was chosen to be the maximum number satisfying equation (\ref{eq:ChoiseBeta}). A calculation showing that equation (\ref{eq:ChoiseBeta}) holds is deferred to \Cref{app:ChooseBeta}.
	Consider the ball $B_w=B(\vec{w},\|\vec{t}-\vec{w}\|)$. $B_w$ has radius at most $\frac1\beta$, and the segment $[\vec{t},\vec{t'}]$ is a diameter of $B_w$. 
	Furthermore, $[\vec{t},\vec{t'}]$ is a chord in both $B_p$ and $B_q$.
	
	\begin{wrapfigure}{r}{0.27\textwidth}
		\begin{center}
			\vspace{-25pt}
			\includegraphics[scale=0.9]{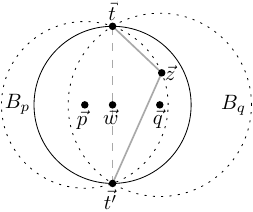}
			\vspace{-5pt}
		\end{center}
		\vspace{-10pt}
	\end{wrapfigure}
	Recall that we assume that $\vec{z}=(z_x,z_y)\in B_p\cap B_q$. If $z_x\ge w$, then the chord $[\vec{t},\vec{t'}]$ of $B_p$ separates the point $\vec{z}$ from the center $\vec{p}$, because $0<w<q$ (see illustration on the right).
	It follows that the angle $\angle \vec{t}\vec{z}\vec{t'}$ is larger than $\frac\pi2$, which implies that $\vec{z}\in B_w$ (as $[\vec{t},\vec{t'}]$ is a diameter, for any point $\vec{z}\notin B_w$, the angle $\angle \vec{t}\vec{z}\vec{t'}$ is smaller than $\frac\pi2$).
	If the $z_x<w$, a symmetric argument (using $B_q$) will imply that $\vec{z}\in B_w$. 
	We conclude that in any case $\vec{z}\in B_w$. By equation (\ref{eq:ChoiseBeta}), it follows that  $\|\vec{z}-\vec{w}\|\le\frac1\beta$.\qed
	
	\begin{wrapfigure}{r}{0.27\textwidth}
		\begin{center}
			\vspace{-20pt}
			\includegraphics[scale=0.65]{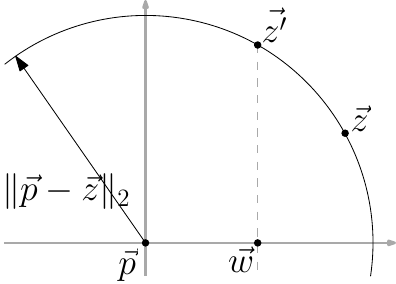}
			\vspace{-5pt}
		\end{center}
		\vspace{-10pt}
	\end{wrapfigure}
	\vspace{8pt}\emph{Proof of \Cref{clm:Bp}.}~
	Assume that $\vec{z}=(z_x,z_y)\notin B_p\cap B_q$.
	Recall that  $z=\|\vec{z}-\vec{w}\|$.
	We show that if $z_x\ge w$ then $\|\vec{z}-\vec{p}\|\ge1+\beta z$, and otherwise $\|\vec{z}-\vec{q}\|\ge\beta q+\beta z$.\\
	First, consider the case when $z_x\ge w$.
	Notice that $\vec{z}\notin B_p$, because the boundaries of $B_p$ and $B_q$ intersect only at $\vec{t},\vec{t'}$, and thus the intersection of $B_p$ with the half plane $x\ge w$ is contained in $B_q$.
	Let $\vec{z'}=(z'_x,z'_y)$ be a point on the ball with radius $\|\vec{z}-\vec{p}\|$ around $\vec{p}$ such that $z'_x=w$ and $z'_y\ge 0$, and notice that $z'_y\ge z_y$ (see illustration on the right).
	Notice that $\|\vec{z'}-\vec{w}\|\ge \|\vec{z}-\vec{w}\|$, because 
	$z_x^2+z_y^2=\|\vec{z}-\vec{p}\|^2=\|\vec{z'}-\vec{p}\|^2=
	w^2+{z'_y}^2$ and $z_x\ge w$, so we get $$\|\vec{z}-\vec{w}\|^2=z_y^2+(z_x-w)^2=z_y^2+z_x^2-2wz_x+w^2=2w^2-2wz_x+{z'_y}^2\le {z'_y}^2=\|\vec{z'}-\vec{w}\|^2~.$$
	
	Since $\|\vec{z}-\vec{p}\|=\|\vec{z'}-\vec{p}\|$, 
	it is enough to show that $\|\vec{z'}-\vec{p}\|\ge1+\beta \|\vec{z'}-\vec{w}\|$. From here on, we will abuse notation and refer to $z'$ as $z$. Thus we simply assume $\vec{z}=(w,z)$.
	
	As $B_p$ and $B_q$ intersect at $\vec{t}$, and $\vec{z}\notin B_p\cap B_q$, it must hold that $z\ge \sqrt{4-w^2}$. Note that $\|\vec{p}-\vec{t}\|=2$ (because $\vec{t}$ is on the boundary of $B_p$), and by equation (\ref{eq:ChoiseBeta}), $\beta\cdot\|\vec{t}-\vec{w}\|\le 1$.
	It thus follows that $1+\beta\|\vec{w}-\vec{t}\|\le2=\|\vec{p}-\vec{t}\|$, implying that the claim holds for $\vec{z}=\vec{t}$.
	It remains to prove that the claim holds for $\vec{z}=(w,\sqrt{4-w^{2}}+\delta)$ for all $\delta\ge 0$.
	It holds that 
	\begin{align*}
	\|\vec{z}-\vec{p}\|^{2} & =w^{2}+(\sqrt{4-w^{2}}+\delta)^{2}=\|\vec{t}-\vec{p}\|^{2}+2\delta\sqrt{4-w^{2}}+\delta^{2}~.\\
	\left(1+\beta\cdot\|\vec{z}-\vec{w}\|\right)^{2} &= \left(1+\beta\cdot\|\vec{t}-\vec{w}\|+\beta\cdot\|\vec{z}-\vec{t}\|\right)^{2} \\ & =\left(1+\beta\cdot\|\vec{t}-\vec{w}\|\right)^{2}+2\beta\|\vec{z}-\vec{t}\|\left(1+\beta\cdot\|\vec{t}-\vec{w}\|\right)+\beta^{2}\|\vec{z}-\vec{t}\|^{2}\\
	& =\left(1+\beta\cdot\|\vec{t}-\vec{w}\|\right)^{2}+2\beta\delta\left(1+\beta\sqrt{4-w^{2}}\right)+\beta^{2}\delta^{2}~.
	\end{align*}
	As $1+\beta\|\vec{w}-\vec{t}\|\le\|\vec{p}-\vec{t}\|$,
	it holds that 
	\begin{align*}
	\|\vec{z}-\vec{p}\|^{2}-\left(1+\beta\cdot\|\vec{z}-\vec{w}\|\right)^{2} &
	\ge\left(2\delta\sqrt{4-w^{2}}+\delta^{2}\right)-\left(2\beta\delta\left(1+\beta\sqrt{4-w^{2}}\right)+\beta^{2}\delta^{2}\right)\\
	& =2\delta\sqrt{4-w^{2}}\left(1-\beta^{2}\right)+\delta^{2}(1-\beta^{2})-2\beta\delta\ge0~,
	\end{align*}
	where the last inequality holds\footnote{See calculation \href{https://www.wolframalpha.com/input/?i=\%284-\%28\%281\%2F2\%29\%281-b\%5E2\%29q-b\%2B3\%2F\%282q\%29\%29\%5E2\%29\%5E\%281\%2F2\%29\%281-b\%5E2\%29\%3E\%3Db\%3B+b\%3D0.55701571}{here}.} as by our choice of $\beta$, we have $\sqrt{4-w^{2}}\left(1-\beta^{2}\right)\ge\beta$ for every $\frac1\beta<q<\frac{1}{1-\beta}$.
	The claim follows. 
	
	Next, we show that in the symmetric case, when $z_x\le w$, it holds that $\|\vec{z}-\vec{q}\|\ge\beta q+\beta z$.
	Similarly to the previous case, we can assume that $\vec{z}=(w,z)$, where $z\ge\sqrt{4-w^2}$ (as this is only harder).
	Now, as $\vec{t}$ lies on the boundary of $B_q$, by equation (\ref{eq:ChoiseBeta}), it holds that $\|\vec{t}-\vec{q}\|=1+\beta q\ge \beta \|\vec{w}-\vec{t}\|+\beta q$.
	It remains to prove that the claim holds for $\vec{z}=(w,\sqrt{4-w^{2}}+\delta)$ for some $\delta>0$.
	It holds that 	
	\begin{align*}
	\|\vec{z}-\vec{q}\|^{2} & =\left(q-w\right)^{2}+(\sqrt{4-w^{2}}+\delta)^{2}=\|\vec{t}-\vec{q}\|^{2}+2\delta\sqrt{4-w^{2}}+\delta^{2}~.\\
	\left(\beta q+\beta\cdot\|\vec{z}-\vec{w}\|\right)^{2} & =
	\left(\beta q+\beta\cdot\|\vec{t}-\vec{w}\|+\beta\cdot\|\vec{z}-\vec{t}\|\right)^{2} \\
	&=\left(\beta q+\beta\cdot\|\vec{t}-\vec{w}\|\right)^{2}+2\beta\|\vec{z}-\vec{t}\|\left(\beta q+\beta\cdot\|\vec{t}-\vec{w}\|\right)+\beta^{2}\|\vec{z}-\vec{t}\|^{2}\\
	& \ge\left(\beta q+\beta\|\vec{t}-\vec{w}\|\right)^{2}+2\beta\delta\left(\beta q+\beta\sqrt{4-w^{2}}\right)+\beta^{2}\delta^{2}~.
	\end{align*}
	Thus, 
	\begin{align*}
	\|\vec{z}-\vec{q}\|^{2}-\left(\beta q+\beta\cdot\|\vec{z}-\vec{w}\|\right)^{2} & \ge\left(2\delta\sqrt{4-w^{2}}+\delta^{2}\right)-\left(2\beta\delta\left(\beta q+\beta\sqrt{4-w^{2}}\right)+\beta^{2}\delta^{2}\right)\\
	& =2\delta\sqrt{4-w^{2}}\left(1-\beta^{2}\right)+\delta^{2}(1-\beta^{2})-2\beta^{2}q\delta\ge0~,
	\end{align*}
	where the last inequality holds\footnote{See calculation \href{https://www.wolframalpha.com/input/?i=\%284-\%28\%281\%2F2\%29\%281-b\%5E2\%29q-b\%2B3\%2F\%282q\%29\%29\%5E2\%29\%5E\%281\%2F2\%29\%281-b\%5E2\%29\%3E\%3Db\%5E2q\%3B+b\%3D0.55701571}{here}.} as by our choice of $\beta$, we have $\sqrt{4-w^{2}}\left(1-\beta^{2}\right)\ge\beta^2q$ for every $\frac1\beta<q<\frac{1}{1-\beta}$.
	The claim follows.\qed
\newline

	\begin{wrapfigure}{r}{0.4\textwidth}
		\begin{center}
			\vspace{-10pt}
			\includegraphics[scale=0.55]{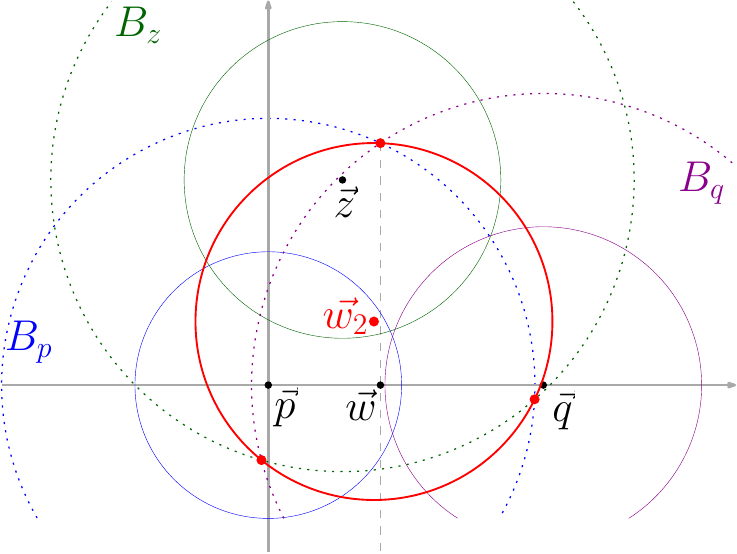}
			\vspace{-5pt}
		\end{center}
		\vspace{-20pt}
	\end{wrapfigure}
\textbf{Remark 1.}
Our proof of \Cref{thm:Euclidean} is based on a simple algorithm: choose a point $\vec{p}\in\mathbb{R}^d$ which minimizes $R_p$. If it is a $\beta$-plurality point - we are done. Otherwise, there is a ball centered at a point $q$ such that the ball of radius $\beta\cdot \|\vec{p}-\vec{q}\|$ contains more than $\frac n2$ voters. We then argue that a choice $\vec{w}$, which is a linear combination of $\vec{p}$ and $\vec{q}$ is a $\beta$-plurality winner.
This algorithm can be naturally extended for another step. Fix some $\beta'>\beta$, and suppose that $\vec{w}$ is not a $\beta'$-plurality point.
In particular, there is a point $\vec{z}$ such that the ball $B(\vec{z},\beta'\|\vec{z}-\vec{w}\|)$ contains more than $\frac n2$ voter points.
We then can hope to find a new choice point $\vec{w_2}$ that will be a $\beta'$-plurality point. Here a natural choice of $\vec{w_2}$ will be the center of the minimal ball containing the intersection of the three balls $B_p=B(\vec{p},2)$, $B_q=B(\vec{q},\beta'\|\vec{q}-\vec{p}\|+1)$, and $B_z=B(\vec{z},\beta'\|\vec{z}-\vec{w}\|+1)$. See illustration on the right.
Even though it is indeed possible that this approach will provide some improvement, it is unlikely to be very significant. The reason is that even for the simplest symmetric case where $\vec{q}=(\frac{1}{\beta'},0)$, $\vec{z}=(\frac{1}{2\beta'},\frac{1}{\beta'})$, one need $\beta'\le \sqrt{\frac{89}{256}}\approx 0.59$. For the hardest case, it is likely that a much smaller $\beta'$ will be required.

	\section{Conclusion}
	Let ${\beta^*=\inf\left\{\beta^*_{(X,d)}~\mid~(X,d)~\mbox{is a metric space}\right\}}$. 
	In this paper we showed that $\sqrt{2}-1\le \beta^*\le \frac12$. Further, in the Euclidean case, for arbitrary dimension $d\ge 4$, by combining our results with \cite{ABGH21}, we know that ${0.557<\beta^*_{(\R^d,\|\cdot\|_2)}\le\frac{\sqrt{3}}{2}}$.
	The main question left open is closing these two gaps.
	Our conjecture is that the upper bounds are tight, since when $|V|=3$, a plurality point must ``win'' $\frac23$ of the overall vote. This task can only become easier once the number of voters increase.
	\begin{conjecture}
		$\beta^*=\frac12$, and $\beta^*_{(\R^d,\|\cdot\|_2)}=\frac{\sqrt{3}}{2}$ for every $d\ge 2$.
\end{conjecture}

If indeed $\beta^*_{(\R^d,\|\cdot\|_2)}=\frac{\sqrt{3}}{2}\approx 0.866$ for every dimension $d$, then it implies that the concept of $\beta$-plurality might be very useful as a relaxation for Condorcet winner. Informally, it shows that the amount of ``compromise'' that we need to make in order to find a plurality point in any Euclidean space is relatively small.

\paragraph{Acknowledgments.}
After sharing our proof of \Cref{thm:main} with the authors of \cite{ABGH21}, Mark de Berg proved a weaker version of \Cref{thm:Euclidean}, and  generously allowed us to publish our proof which is based on his observation. Specifically, de Berg proved that $\beta^*_{(\R^d,\|\cdot\|_2)}\ge\frac12$.\footnote{Following the lines of the proof of \Cref{thm:Euclidean}, for $\beta=\frac12$, the point $\vec{p}$ is a $\frac12$-plurality point, as
no $q$ satisfies equation (\ref{eq:pq}).
} \\
The authors would also like to thank Boris Aronov, Nimrod Talmon, and an anonymous reviewer, for providing useful comments on the manuscript.

{\small
\bibliographystyle{alphaurlinit}
\bibliography{refs}
}
\appendix
\section{Equivalence between the definitions of $\beta$-plurality point}\label{sec:defEquiv}
\begin{lemma}
\Cref{def:Plurpoint} for $\beta(p,V)$ is equivalent to the definition from Aronov \etal \cite{ABGH21}.
\end{lemma}
\begin{proof}
We will use $\beta(p,V)$ to denote the definition given in \cite{ABGH21} (and in our introduction), and  $\tilde{\beta}(p,V)$ to denote our definition from \Cref{def:Plurpoint}. We will show that for every metric space $(X,d)$, voter multiset $V$ in $X$, and point $p\in X$, it holds that $\beta(p,V)=\tilde{\beta}(p,V)$. The equivalence between the other parameters will follow.
Fix $|V|=n$. There are two directions for the proof:
\begin{itemize}
	\item \sloppy $\beta(p,V)\le\tilde{\beta}(p,V)$. Assume by contradiction that $\tilde{\beta}(p,V)<\beta(p,V)$, so thus there exists some $\alpha\in (\tilde{\beta}(p,V),\beta(p,V)]$.
	By the definition of $\beta(p,V)$, for every $q\in X$ it holds that
	$\left|\left\{ v\in V\mid\alpha\cdot d(p,v)<d(q,v)\right\} \right|\ge\left|\left\{ v\in V\mid\alpha\cdot d(p,v)>d(q,v)\right\} \right|$. Clearly, for the weak inequality we get $\left|\left\{ v\in V \mid\alpha\cdot d(p,v)\le d(q,v)\right\} \right|\ge\frac{n}{2}$, and thus $\tilde{\beta}(p,V)\ge\alpha$, a contradiction.
	
	\item \sloppy $\tilde{\beta}(p,V)\le\beta(p,V)$. Assume by contradiction that  $\beta(p,V)<\tilde{\beta}(p,V)$, so there exists $\varepsilon>0$ such that $\beta(p,V)+\varepsilon<\tilde{\beta}(p,V)$.
	By the definition of $\tilde{\beta}(p,V)$, there exists
	$\alpha\ge\beta(p,V)+\varepsilon$ such that for every $q$, we have $\left|\left\{ v\in V\mid\alpha\cdot d(p,v)\le d(q,v)\right\} \right|\ge\frac{n}{2}$.
	Let $\alpha'=\alpha-\frac{\varepsilon}{2}\in(\beta(p,V),\alpha)$. Then for every $q\ne p$, we have $\left|\left\{ v\in V\mid\alpha'\cdot d(p,v)<(q,v)\right\} \right|\ge\left|\left\{ v\in V\mid\alpha\cdot d(p,v)\le d(q,v)\right\} \right|\ge\frac{n}{2}$,
	implying that $\left|\left\{ v\in V\mid\alpha'\cdot d(p,v)<d(q,v)\right\} \right|\ge\left|\left\{ v\in V\mid\alpha'\cdot d(p,v)>d(q,v)\right\} \right|$.
	Clearly, for $q=p$, it holds that $\left|\left\{ v\in V\mid\alpha'\cdot d(p,v)<d(q,v)\right\} \right|\ge\left|\left\{ v\in V\mid\alpha'\cdot d(p,v)>d(q,v)\right\} \right|$.
	It follows that $p$ is an $\alpha'$-plurality point, a contradiction. 
\end{itemize}
\end{proof}

\section{Proof of \Cref{clm:ballObservation}}\label{appballObservation}
\begin{proof}
	By translation and rotation, we can assume w.l.o.g. that $\vec{a}=\vec{0}$, and $\vec{b}=\|\vec{a}-\vec{b}\|\cdot e_1$ ($e_1$ here is the first standard basis vector).
	A straightforward calculation shows that  
	\begin{align*}
	&\left\{ \vec{x}\in\R^{d}\mid\beta\cdot\left\| \vec{a}-\vec{x}\right\| >\| \vec{b}-\vec{x}\| \right\}  \\&\qquad =\left\{ \vec{x}\in\R^{d}\mid\left(x_{1}-\| \vec{a}-\vec{b}\| \right)^{2}+\sum_{i=2}^d x_{i}^{2}<\beta^{2}\cdot\sum_{i= 1}^d x_{i}^{2}\right\} \\
	&\qquad =\left\{ \vec{x}\in\R^{d}\mid\left(1-\beta^2\right)x_{1}^2-2x_1\| \vec{a}-\vec{b}\|+\| \vec{a}-\vec{b}\|^2+\left(1-\beta^2\right)\sum_{i=2}^d x_{i}^{2}<0\right\} \\
	&\qquad =\left\{ \vec{x}\in\R^{d}\mid\left(x_{1}-\frac{\| \vec{a}-\vec{b}\| }{1-\beta^{2}}\right)^{2}+\sum_{i=2}^d x_{i}^{2}<\frac{\beta^{2}\| \vec{a}-\vec{b}\| ^{2}}{\left(1-\beta^{2}\right)^{2}}\right\}~. 
	\end{align*}
	Thus we indeed obtain a ball with center at $\vec{o}=\frac{\| \vec{a}-\vec{b}\| }{1-\beta^{2}}\cdot e_{1}=\vec{a}+\frac{1}{1-\beta^{2}}\cdot(\vec{a}-\vec{b})$, and radius $r=\sqrt{\frac{\beta^{2}\| \vec{a}-\vec{b}\| ^{2}}{\left(1-\beta^{2}\right)^{2}}}=\beta\cdot\left\| \vec{o}-\vec{a}\right\|$.
\end{proof}

\section{Proof of equation (\ref{eq:ChoiseBeta})}\label{app:ChooseBeta}
Set $f(\beta,q)=\|\vec{t}-\vec{w}\|^2=4-w^{2}=4-\left(\frac{1}{2}(1-\beta^{2})q-\beta+\frac{3}{2q}\right)^{2}$.
We will show that for our choice of $\beta$, and for every $q\in(\frac1\beta,\frac{1}{1-\beta})$, it holds that $\sqrt{f(\beta,q)}\le\frac{1}{\beta}$, thus proving equation (\ref{eq:ChoiseBeta}). We have
\[
\frac{\partial}{\partial q}f(\beta,q)=2\left(\frac{1}{2}(1-\beta^{2})q-\beta+\frac{3}{2q}\right)\left(\frac{1}{2}(1-\beta^{2})-\frac{3}{2q^{2}}\right)
\]
which equals to $0$ only for $q\in\left\{ \pm\sqrt{\frac{3}{1-\beta^{2}}},\frac{\sqrt{4\beta^{2}-3}\pm\beta}{\beta^{2}-1}\right\} $.\footnote{See calculation \href{https://www.wolframalpha.com/input/?i=\%284\%E2\%88\%92\%281\%2F2\%281\%E2\%88\%92b\%5E2\%29x\%E2\%88\%92b\%2B3\%2F\%282x\%29\%29\%5E2\%29\%27}{here}.}

As we restrict our attention to $q\in(\frac1\beta,\frac{1}{1-\beta})$, it follows that
once we fixed $\beta$, the functoin $f(\beta,q)$ has a maximum at $\sqrt{\frac{3}{1-\beta^{2}}}$ (note that $\sqrt{\frac{3}{1-\beta^{2}}}\in(\frac1b,\frac{1}{1-b})$ for every $b\in(\frac12,1)$).
It thus will be enough to prove that 
\[
f(\beta,q)\le f\left(\beta,\sqrt{\frac{3}{1-\beta^{2}}}\right)=1+2\beta^{2}+2\sqrt{3}\sqrt{1-\beta^{2}}\beta\le\frac{1}{\beta^{2}}\,.
\]
This expression could be ``massaged'' into a degree 4 polynomial.
Thus we can obtain an exact solution. In particular, for every $\beta\in\left(0,\frac{1}{2}\sqrt{\frac{1}{2}+\sqrt{3}-\frac{1}{2}\sqrt{4\sqrt{3}-3}}\right]\approx(0,0.557]$,\footnote{See calculation \href{https://www.wolframalpha.com/input/?i=1\%2B2b\%5E2\%2B2*sqrt\%283\%29sqrt\%281-b\%5E2\%29b+\%3C\%3D+1\%2F\%28b\%5E2\%29}{here}.}
it holds that $\sqrt{f(\beta,q)}\le\frac{1}{\beta}$, as required.

\end{document}